\documentclass[11pt,a4paper]{llncs}
\usepackage{amssymb}
\usepackage{url}
\usepackage{algorithm}
\usepackage{algorithmic}
\usepackage{epsfig}
\usepackage{graphicx}
\usepackage{epstopdf}
\usepackage{verbatim}
\newtheorem{thm}{Theorem}
\newtheorem{lem}[thm]{Lemma}

\newcounter{ALC@tempcntr}
\newcommand{\LCOMMENT}[1]{%
    \setcounter{ALC@tempcntr}{\arabic{ALC@rem}}
    \setcounter{ALC@rem}{1}
    \item \{#1\}
    \setcounter{ALC@rem}{\arabic{ALC@tempcntr}}
}%
\makeatother

\urldef{\mailsa}\path|cs14resch01002@iith.ac.in|

\begin{document}

\title{A Simple Algorithm For Replacement Paths Problem}
\title{A Simple Algorithm For Replacement Paths Problem}
\author{Anjeneya Swami Kare}
\institute{Department of Computer Science and Engineering\\
Indian Institute of Technology\\ Hyderabad, India\\ \mailsa}

\maketitle
\begin{abstract}
 Let $G = (V, E)$ ($|V| = n$ and $|E| = m$) be an undirected graph with positive edge weights.
 Let $P_{G}(s, t)$ be a shortest $s-t$ path in $G$. Let $l$ be the number of edges in $P_{G}(s, t)$.
 The \emph{Edge Replacement Path}  problem is to compute a shortest $s-t$ path in $G\backslash\{e\}$, for every edge
 $e$ in $P_{G}(s, t)$. The \emph{Node Replacement Path} problem is to compute a shortest
 $s-t$ path in $G\backslash\{v\}$, for every vertex $v$ in $P_{G}(s, t)$.

 In this paper we present an $O(T_{SPT}(G)+m+l^2)$ time and $O(m+l^2)$ space algorithm
 for both the problems. Where, $T_{SPT}(G)$ is the asymptotic time to compute a single source
 shortest path tree in $G$. The proposed algorithm is simple and easy to implement.
 \end{abstract}


\section{Introduction}
Let $G = (V, E)$ ($|V| = n$ and $|E| = m$) be an undirected graph with a weight
function $w : E \rightarrow \mathbb{R}_{>0}$ on the edges.
Let $P_{G}(s, t) = \{v_{0}=s, v_{1}, \ldots , v_{l-1}, v_{l}=t\}$ be a shortest
$s-t$ path in $G$. Let $l$ denote the  number of edges in $P_{G}(s, t)$, also denoted by $|P_{G}(s, t)|$.
The total weight of the path $P_{G}(s, t)$ is denoted by $d_{G}(s, t)$,
i.e, $d_{G}(s, t) = \sum_{i = 1}^{l}{w(e_{i})}$, where, $e_{i}$ is the
edge $(v_{i-1}, v_{i}) \in P_{G}(s, t)$. A shortest path tree (SPT) of $G$ rooted at $s$ 
(respectively, $t$) is denoted by $T_{s}$ (respectively, $T_{t}$).

A \emph{Replacement Shortest Path} (RSP) for the edge $e_{i}$
(respectively, node $v_{i}$) is a shortest $s-t$ path in $G\backslash{e_{i}}$
(respectively, $G\backslash{v_{i}}$). The \emph{Edge Replacement Path} problem is to compute
RSP for all $e_{i} \in P_{G}(s, t)$. Similarly, the \emph{Node Replacement Path}
problem is to compute RSP for all $v_{i} \in P_{G}(s, t)$. 

Like in all existing algorithms for RSP problem, our algorithm has two phases:
\begin{enumerate}
  \item Computing shortest path trees rooted at $s$ and $t$, $T_{s}$ and $T_{t}$ respectively.
  \item Computing RSP using $T_{s}$ and $T_{t}$.
\end{enumerate}
For graphs with non-negative edge weights, computing
an SPT takes $O(m + n \log n)$ time, using the standard Dijkstra's algorithm~\cite{dijkstra} with
Fibonacci heaps~\cite{fheaps}. However, for integer weighted graphs (RAM model)~\cite{thorup},
planar graphs~\cite{planarspt} and minor-closed graphs~\cite{minorspt},
$O(m+n)$ time algorithms are known. In this paper, to compute SPTs $T_{s}$ and $T_{t}$ (phase 1)
we use the existing algorithms. For phase 2, we present an $O(m + l^2)$ time algorithm
which is simple and easy to implement.

Motivation for studying replacement paths problem is its relevance in single link (or node) recovery protocols.
Other problems which are closely related to replacement paths problem are \emph{Most Vital Edge} problem~\cite{fastermve},
\emph{Most Vital Node} problem~\cite{mostvn} and \emph{Vickrey Pricing}~\cite{vickreyprice}.
Often an algorithm for replacement paths problem is used as a subroutine in finding $k$-simple
shortest paths between a pair of nodes.

\begin{table}[H]
\centering
\begin{tabular}{|c|c|}
  \hline
  \multicolumn{2}{|c|}{\textbf{Edge Replacement Path Problem}}  \\
  \hline
  Reference & Time Complexity  \\
  \hline
  Malik et al.~\cite{kmve} (1989) & $O(T_{SPT}(G)+ m + n \log n)$  \\
  \hline
  Hershberger and Suri~\cite{vickreyprice} (1997) & $O(T_{SPT}(G)+ m + n \log n)$  \\
  \hline
  Nardelli et al.~\cite{fastermve} (2001) & $O(T_{SPT}(G)+ m \alpha(m, n))$  \\
  \hline
  Jay and Saxena~\cite{jay} (2013) & $O(T_{SPT}(G)+ m + d^2)$  \\
  \hline
  Lee and Lu~\cite{linearrsp} (2014) & $O(T_{SPT}(G)+ m + n)$  \\
  \hline
  This Paper & $O(T_{SPT}(G)+ m + l^2)$  \\
  \hline
  \multicolumn{2}{|c|}{\textbf{Node Replacement Path Problem}}  \\
  \hline
  Reference & Time Complexity  \\
  \hline
  Nardelli et al.~\cite{mostvn} (2003) & $O(T_{SPT}(G)+ m + n \log n)$  \\
  \hline
  Jay and Saxena~\cite{jay} (2013) & $O(T_{SPT}(G)+ m + d^2)$  \\
  \hline
  Lee and Lu~\cite{linearrsp} (2014) & $O(T_{SPT}(G)+ m + n)$  \\
  \hline
  Kare and Saxena~\cite{asksax} (2014) & $O(T_{SPT}(G)+ m \alpha(m, n))$  \\
  \hline
  This Paper & $O(T_{SPT}(G)+ m + l^2)$  \\
  \hline
  \end{tabular}
\label{table1}
\caption{Summary of existing algorithms \protect\footnotemark for RSP problem}
\end{table}

\footnotetext{In the referenced papers, authors ignore the term $T_{SPT}(G)$, as they assume
either shortest path trees are given or restriction on the input graph class for which
linear time algorithms are known for SPT}

For the Edge Replacement Path problem Malik et al.~\cite{kmve} and Hershberger
and Suri~\cite{vickreyprice} independently gave $O(T_{SPT}(G)+ m + n \log n)$
time algorithms. Nardelli et al.~\cite{fastermve} gave an $O(T_{SPT}(G)+ m \alpha(m, n))$
time algorithm, where $\alpha$ is the inverse Ackermann function.

For the Node Replacement Path problem Nardelli et al.~\cite{mostvn} gave an $O(T_{SPT}(G)+ m + n \log n)$
time algorithm. Kare and Saxena~\cite{asksax} gave an $O(T_{SPT}(G)+ m \alpha(m, n))$ time algorithm.

Jay and Saxena~\cite{jay} gave an $O(T_{SPT}(G)+ m + d^2)$ algorithm, where $d$ is the diameter of the graph.
Their algorithm can be used to solve both the edge and the node
replacement path problems. They used linear time algorithms for Range Minima Query
(RMQ)~\cite{rmq} and integer sorting in their solution. A total of $2l$ instances, each of
RMQ and integer sorting has been used (with size of each instance at most $l$).
Recently, Lee and Lu~\cite{linearrsp} gave an $O(T_{SPT}(G)+ m + n)$ time algorithm.
Table~\ref{table1} summarises the existing algorithms for RSP problem.

In this paper, we present an $O(T_{SPT}(G)+ m + l^2)$ time and $O(m + l^2)$ space algorithm.
The asymptotic complexity of our algorithm matches that of~\cite{jay}. However, our solution does not use RMQ
and integer sorting. Our algorithm organizes the non-tree edges of the graph in a simple manner.
Moreover, an advantage of our algorithm over~\cite{jay} and~\cite{linearrsp} is that, in a single iteration both
edge and node replacement paths can be obtained, whereas in~\cite{jay} and~\cite{linearrsp} the algorithm has to
be run independently for the edge and node replacement paths. Note that, linear time algorithm for RMQ~\cite{rmq}
and the algorithm in~\cite{linearrsp} are complex to implement. The simplicity of our algorithm
makes it an ideal candidate for the RSP. In particular, for dense graphs and graphs with small diameter
($l \leq diameter(G) = O(\sqrt m$)) our algorithm is optimal and matches with that of~\cite{linearrsp}.
As observed in~\cite{jay}, graphs in real world data sets have small diameter, which further adds significance
to our algorithm.

The contribution of this paper is summarized in the following theorem.

\begin{thm}
\label{thm1}
There is an algorithm for the edge and the node replacement path problems that runs in
$O(T_{SPT}(G)+ m + l^2)$ time using $O(m + l^2)$ space.
\end{thm}

The rest of the report is organised as follows. In Section~\ref{ersp} we discuss the
algorithm for the edge replacement path problem. In Seection~\ref{nrsp} we discuss the
algorithm for node replacement path problem. We conclude with Section~\ref{conc}.
\section{Edge Replacement Paths}
\label{ersp}
We start by computing shortest path trees $T_{s}$ and $T_{t}$. In the rest of the section
we describe the algorithm for computing RSP using $T_{s}$ and $T_{t}$ (phase 2).

A potential replacement path for the edge $e_{i} = (v_{i-1}, v_{i})$ can be
seen as a concatenation of three paths $A$, $B$ and $C$, where, $A = s \leadsto v_{k} \in P_{G}(s, t)$ for
some $0 \leq k < i$, $B = v_{k} \leadsto v_{r} \in G \backslash E(P_{G}(s, t))$  for some $i \leq r \leq l$ and
$C = v_{r} \leadsto t \in P_{G}(s, t)$ as shown in \figurename~\ref{fig1}. Here, the symbol $\leadsto$ represents a path in $G$.
One extreme case is when $|A| = 0$ and $|C| = 0$ (i. e. $v_{k} = s$ and $v_{r} = t$) as shown in \figurename~\ref{fig1}(b).
Such a replacement path is also a potential replacement path for all the edges $e_{i} \in P_{G}(s, t)$. The other extreme case is when $|A| = i - 1$
and $|C| = l - i$ (i.e. $v_{k} = v_{i-1}$ and $v_{r} = v_{i}$) as shown in \figurename~\ref{fig1}(c).
Such a replacement path is a potential replacement path only for the edge $e_{i}$.

\begin{figure}[!ht]
\centering
\includegraphics[width=4.0in]{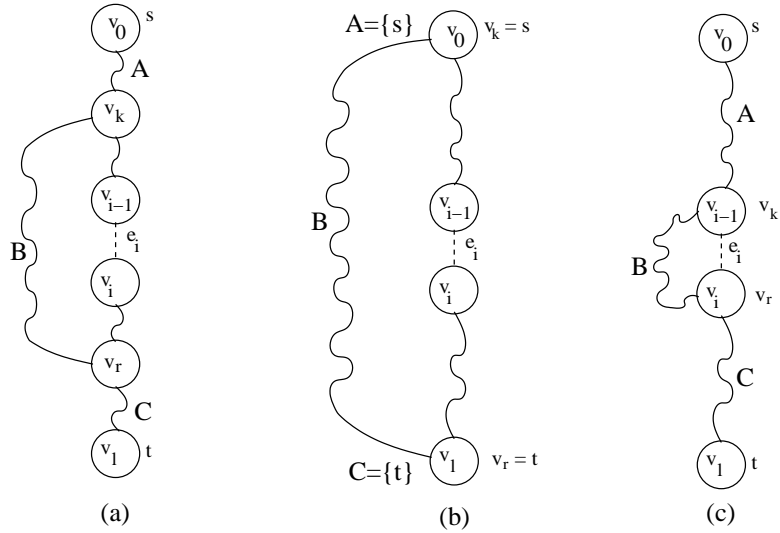}
\caption{Potential replacement paths for the edge $e_{i}$.
The zig-zag lines represent a path} \label{fig1}
\end{figure}

Consider the shortest path tree rooted at $s$ ($T_{s}$). When the edge $e_{i} = (v_{i-1},v_{i})$
is removed from $T_{s}$, $T_{s}$ is disconnected into two sub-trees. $T_{1}(e_{i})$ (sub-tree rooted at $s$)
and $T_{2}(e_{i})$ (sub-tree rooted at $v_{i}$). The vertex sets of $T_{1}(e_{i})$ and $T_{2}(e_{i})$
determine a cut in the graph $G$. Let $C(e_{i})$ denote the set of all non-tree edges crossing the cut.
These edges are called crossing edges, i.e,
$C(e_{i}) = \{ (x, y) \in E \backslash {e_{i}} | x \in T_{1}(e_{i}) \wedge y \in T_{2}(e_{i}) \}$.
In order to have a replacement path, the set $C(e_{i})$ needs to be nonempty. And,
any replacement path must use at least one crossing edge from $C(e_{i})$.
Moreover, as we see from the Lemmas~\ref{lemmadist} and~\ref{lemmace} there exists
an RSP that uses exactly one crossing edge.


\begin{lem}
\label{lemmadist}
For all $(x, y) \in C(e_{i})$, $d_{G-{e_{i}}}(s,x)$ = $d_{G}(s,x)$ and $d_{G-{e_{i}}}(y,t)$ = $d_{G}(y,t)$.
\end{lem}

\begin{proof}
If $(x,y) \in C(e_{i})$, then $x \in T_{1}(e_{i})$ and $y \in T_{2}(e_{i})$.
Shortest $s-x$ path is fully in $T_{1}(e_{i})$ and does not include
the edge $e_{i}$. Hence, $d_{G-{e_{i}}}(s,x)$ = $d_{G}(s,x)$.

To prove $d_{G-{e_{i}}}(y,t) = d_{G}(y,t)$, for the sake of contradiction,
let us assume that $d_{G-{e_{i}}}(y,t) \neq d_{G}(y,t)$ (i.e $d_{G-{e_{i}}}(y,t) > d_{G}(y,t)$).
It means that, $P_{G}(y,t)$ uses the edge $e_{i} = (v_{i-1}, v_{i})$.
This implies $d_{G}(v_{i}, y) > d_{G}(v_{i-1}, y)$.
Since $y \in T_{2}(e_{i})$, $d_{G}(v_{i-1}, y) > d_{G}(v_{i}, y)$ a contradiction.
Hence, $d_{G-{e_{i}}}(y,t) = d_{G}(y,t)$.
\qed
\end{proof}

\begin{lem}
\label{lemmace}
For any edge $e_{i} \in P_{G}(s,t)$, there exists a shortest $s-t$ path in $G-{e_{i}}$
which contains exactly one edge from $C(e_{i})$.
\end{lem}

\begin{proof}
Let us consider a shortest $s-t$ path in $G-{e_{i}}$ (say $P_{1}$) which
uses more than one crossing edge from $C(e_{i})$. Let $(x, y)$ be the last crossing edge
in $P_{1}$. Clearly $x \in T_{1}(e_{i})$ and $y \in T_{2}(e_{i})$. By replacing the part of $P_{1}$ from
$s$ to $x$, by the $s \leadsto x$ path in $T_{1}(e_{i})$, we get a new path which is not
longer than $P_{1}$ and uses exactly one edge from $C(e_{i})$.
\qed
\end{proof}

Using the Lemmas~\ref{lemmadist} and~\ref{lemmace}, we write the total weight of the RSP for
the edge $e_{i}$
as:
\begin{equation}
\label{ersplen}
d_{G-e_{i}}(s,t)=\min_{(x',y')\in C(e_{i})} \{d_{G}(s,x')+w(x', y')+d_{G}(y',t)\}.
\end{equation}

All the terms in the equation (\ref{ersplen}) are available in constant time
for a fixed $(x',y')$ from $T_{s}$ and $T_{t}$. Let $(x,y)$ be the crossing edge that minimizes
the RHS of the equation (\ref{ersplen}). We call that $(x,y)$ the \emph{swap edge}. If we have the swap edge, we can report
the RSP as $s \leadsto x \rightarrow y \leadsto t$ in constant time. Every non-tree edge can be a potential crossing edge for every edge in $P_{G}(s, t)$. So, solving equation (\ref{ersplen}) by brute force gives us $O(ml)$ time algorithm. In this paper we present an $O(m + l^2)$ time algorithm.
In the rest of the paper, we concentrate on computing the swap edge for each $e_{i} \in P_{G}(s, t)$.
\subsection{Labeling the nodes of $G$}
\label{labeling}
Every vertex of $G$ is labeled with an integer value from $0$ to $l$, with
respect to the shortest path tree $T_{s}$. The process of labeling is as follows:

Let $T_{v_{i}}$ be the sub-tree rooted at the node $v_{i}$ in $T_{s}$.
All the nodes in the sub-tree $T_{v_{l}}$ are labeled with the integer value $l$.
For $0 \leq i < l$, all the nodes in the sub-tree $T_{v_{i}} \backslash T_{v_{i+1}}$
are labeled with the integer value $i$. See \figurename~\ref{fig2}(a) for an example labeling.

Using pre-order traversal on $T_{s}$, we compute the labels of all the vertices in linear time.
We start pre-order traversal from the source vertex $s$ using zero as initial label.
While visiting the children of a node recursively, the child node part of $P_{G}(s,t)$ (if any)
will be visited last with an incremented label.
Let $label(v)$ denote the label of a vertex $v$ in $G$. The following Lemma is straightforward.

\begin{lem}
\label{lemmalabel1}
A non-tree edge $(x,y) \in C(e_{i})$ if and only if $label(x) < i$ and $label(y) \geq i$.
In other words, for a non-tree edge $(x,y)$, if $label(x) = i$ and $label(y) = i + r$
for some $r > 0$, then $(x,y) \in C(e_{j})$, $\forall (i < j \leq i + r)$.
\end{lem}
\subsection{Computing Swap Edges}
We construct a directed acyclic graph which will aid us in computing the swap edges.
We call this DAG as RSP-DAG, denoted by $\widehat{G}$.
The following algorithm explains the construction of the RSP-DAG.
An example RSP-DAG is shown in \figurename~\ref{fig2}(b).

\begin{algorithm}[H]
    \caption{Algorithm to construct the RSP-DAG $\widehat{G}$ = $(\widehat{V},\widehat{E})$.}
    \label{rspdag}
    \begin{algorithmic}[1]
    \LCOMMENT{Adding Nodes. Each node is identified by an ordered pair $(i,j)$}
    \STATE $\widehat{V} = \emptyset$
    \STATE $\widehat{E} = \emptyset$
    \FOR{$i = 0 \mbox{ to } l-1$}
        \FOR{$j = i+1 \mbox{ to } l$}
           \STATE $\widehat{V} = \widehat{V} \cup (i,j)$
        \ENDFOR
    \ENDFOR
    \LCOMMENT{Adding Edges}
    \FORALL{ $\widehat{u} = (i, j) \in \widehat{V}$}
        \IF{$j - i > 1$}
            \STATE $\widehat{E} = \widehat{E} \cup ((i,j), (i,j-1))$
            \STATE $\widehat{E} = \widehat{E} \cup ((i,j), (i+1,j))$
        \ENDIF
    \ENDFOR
    \end{algorithmic}
\end{algorithm}

Clearly, the number of vertices in the RSP-DAG is $O(l^2)$ and the number of edges is also
$O(l^2)$. Every node has in-degree and out-degree at most two.
The node with identifier $(0,l)$ has zero in-degree. Nodes $(i, i+1), \forall (0 \leq i < l)$
have zero out-degree (sink nodes).

\begin{figure}[!ht]
\centering
\includegraphics[width=4.8in,height=2.3in]{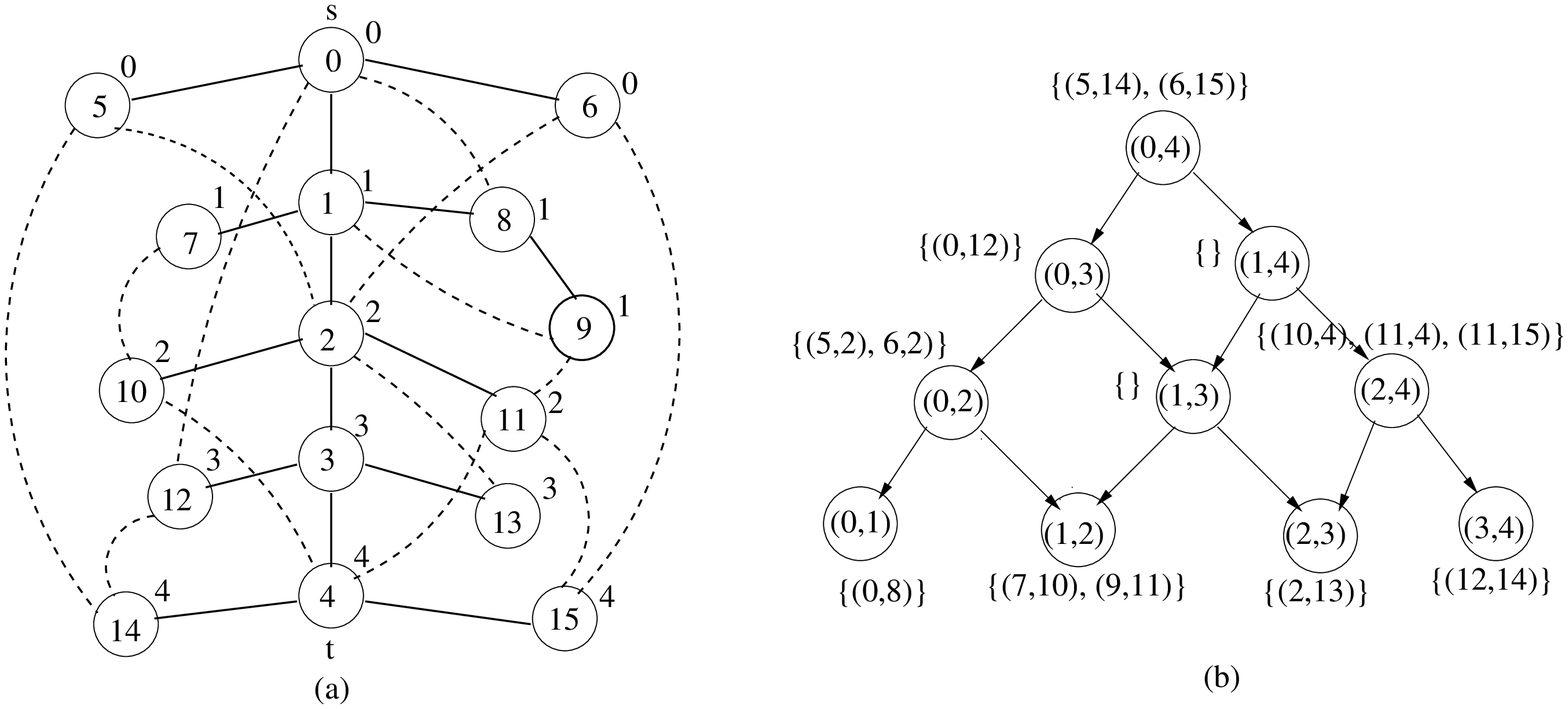}
\caption{(a)An SPT rooted at $s$. Solid lines are part of the SPT. Dashed lines represent
the non-tree edges (we omit the edge weights). Number inside the vertex circle denotes the vertex number, where as
the number above the vertex circle represents vertex label.\newline
(b)Corresponding RSP-DAG with set of non-tree edges associated with nodes} \label{fig2}
\end{figure}

For each node $\widehat{u}=(i,j) \in \widehat{V}$, we associate a set $E_{(i,j)}$ of crossing edges.
This set includes all the non-tree edges $(x,y)$ such that $label(x)=i$ and $label(y)=j$.
This association of crossing edges partitions the crossing edges into disjoint sets.

\begin{lem}
If the swap edge $(x,y)$ for the tree edge $e_{i} \in P_{G}(s,t)$ is present in
the edge set ($E_{(j,k)}$) of a node $\widehat{u} = (j, k) \in \widehat{V}$, then there exists a directed path
from the node $\widehat{u}$ to the node $\widehat{w} = (i-1,i) \in \widehat{V}$ in the RSP-DAG.
\end{lem}
\begin{proof}
Clearly $j \leq i-1$ and $k \geq i$, otherwise, $(x,y)$ will not be the crossing
edge for $e_{i}$. If $\widehat{u}$ is a sink node ($\widehat{u} = \widehat{w}$) in the RSP-DAG, then
the theorem is trivially true.

Otherwise, if we observe the way edges are added in the RSP-DAG, for the node
$\widehat{u} = (j, k) \in \widehat{V}$, two directed edges $((j,k), (j,k-1))$ and $((j,k), (j+1,k))$
are added and from these nodes, we keep adding edges to the lower level nodes in the RSP-DAG. We will
eventually connect to the leaf node $\widehat{w} = (i-1,i) \in \widehat{V}$. Hence there is
a directed path from $\widehat{u}$ to $\widehat{w}$.
\qed
\end{proof}

Now we make a BFS traversal on the RSP-DAG starting from the node with identifier
$(0,l)$. During the traversal, at every node, the minimum cost non-tree edge $(x,y)$
(cost being $d(s,x) + w(x,y) + d(y,t)$) from the corresponding edge set is inserted into the edge sets of its
two children. By the end of this process, minimum cost non-tree edges in the respective
sink nodes give us the swap edges.

\begin{thm}
\label{thm2}
There is an algorithm for
the Edge Replacement Path problem that runs in $O(T_{SPT}(G)+ m + l^2)$ time using $O(m + l^2)$ space.
\end{thm}

\begin{proof}
$T_{SPT}(G)$ represents the time to compute SPTs $T_{s}$ and $T_{t}$.
Construction of the RSP-DAG takes $O(m + l^2)$ time and $O(m + l^2)$ space. BFS traversal on the RSP-DAG
takes $O(l^2)$ time. During the traversal at each node $(i,j) \in \widehat{V}$, we
extract the minimum cost non-tree edge from the set of size at most $|E_{(i,j)}| + 2$.
Time complexity of overall edge extraction steps is:
$\sum_{i<j}|E_{(i,j)}| + 2$ = $O(m + l^2)$. Therefore the total time complexity is $O(T_{SPT}(G)+ m + l^2)$.
Space complexity is $O(m + l^2)$ which is the space to store the RSP-DAG.
\qed
\end{proof}

Using the linear time algorithms for SPT, for integer weighted graphs,
minor closed graphs our algorithm takes $O(m + l^2)$ time.

\section{Node Replacement Paths}
\label{nrsp}
When the node $v_{i} \in P_{G}(s,t)$ is removed,
the SPT $T_{s}$ is partitioned as:
$T_{1}(v_{i})$ (sub-tree rooted at $s$), $T_{2}(v_{i})$ (sub-tree rooted at $v_{i+1}$) and
$F(v_{i})$ (the remaining forest $T_{s} \backslash \{T_{1}(v_{i}) \cup T_{2}(v_{i}) \cup {v_{i}}\}$).
The crossing edges are denoted as:
\begin{eqnarray}
  C'(v_{i}) &=&  \{ (x, y) \in E | x \in T_{1}(v_{i}) \wedge y \in T_{2}(v_{i}) \}  \\
  C''(v_{i}) &=&  \{ (x, y) \in E \backslash (v_{i}, v_{i+1})| x \in F(v_{i}) \wedge y \in T_{2}(v_{i}) \} \\
  C(v_{i}) &=&  C'(v_{i}) \cup C''(v_{i})
\end{eqnarray}

\begin{lem}
\label{lemmadistvrp}
For all $x \in T_{1}(v_{i})$, $d_{G-{v_{i}}}(s,x)$ = $d_{G}(s,x)$, and for all $y \in T_{2}(v_{i})$, $d_{G-{v_{i}}}(y,t)$ = $d_{G}(y,t)$.
\end{lem}

\begin{proof}
We omit the proof as the proof is similar to lemma~\ref{lemmadist}
\qed
\end{proof}

Using Lemma~\ref{lemmadistvrp}, the length of the RSP is written as:
\begin{eqnarray}
  d'_{G-v_{i}}(s,t) &=&  \min_{(x,y) \in C'(v_{i})}\{ d_{G}(s,x) + w(x,y) + d_{G}(y,t) \}  \label{eq:5}\\
  d''_{G-v_{i}}(s,t) &=&  \min_{(x,y) \in C''(v_{i})}\{ d_{G-v_{i}-T_{2}(v_{i})}(s,x) + w(x,y) + d_{G}(y,t) \}\label{eq:6}\\
  d_{G-v_{i}}(s,t) &=& \min \{ d'_{G-v_{i}}(s,t), d''_{G-v_{i}}(s,t)\}
\end{eqnarray}

Having $T_{s}$ and $T_{t}$, all the terms in the equations (\ref{eq:5}) and (\ref{eq:6}) are available in constant time, except the distance $d_{G-v_{i}-T_{2}(v_{i})}(s,x)$ for $x \in F(v_{i})$ (partial shortest path distance).
We need all the partial shortest path distances $d_{G-v_{i}-T_{2}(v_{i})}(s,x)$, $\forall v_{i} \in P_{G}(s,t)$ and $\forall x \in F(v_{i})$.

To compute all the partial shortest path distances, we use the technique used in~\cite{jay} and~\cite{linearrsp}.

Let $G_{i}$ (corresponding to the vertex $v_{i}$) be the graph constructed from $G$ as follows:
The vertex set of $G_{i}$, $V(G_{i})$, consists of the source vertex $s$ and the vertices which are
part of the forest $F(v_{i})$. The edge set of $G_{i}$, $E(G_{i})$, consists of the following edges:

\begin{itemize}
  \item Edges between the nodes within the forest $F(v_{i})$.
        These edges will get the same edge weight as in $G$.
  \item For every $v \in F(v_{i})$, an edge $(s, v)$ is added whenever there
       is at least one edge from $T_{1}(v_{i})$ to $v$.
       The weight of this edge is calculated as follows:
       \begin{equation}
       \widetilde{w}(s,v) = \min_{(u,v)\in E(T_{1}(v_{i}), v))} \{d_{G}(s,u)+w(u, v)\}
       \end{equation}
  \end{itemize}
That is,
\begin{eqnarray}
  V(G_{i}) &=&  \{ V(F(v_{i}))\} \cup \{s\} \\
  E(G_{i}) &=&  \{ E(F(v_{i}), F(v_{i}))\} \cup \{  (s, v) | (v \in F(v_{i}) \wedge E(T_{1}(v_{i}), v) \neq \emptyset) \}
\end{eqnarray}

$G_{i}$ is a graph minor of $G$, since it can be obtained by edge contraction.
Hence, SPT, $T_{i}(s)$ of $G_{i}$, rooted at $s$ can be constructed in $T_{SPT}(G_{i})$ time.
Moreover, $d_{G-v_{i}-T_{2}(v_{i})}(s,x) = d_{G_{i}}(s, x)$ for any $x \in F(v_{i})$.
As $F(v_{i}) \cap F(v_{j}) = \emptyset$, for any $i \neq j$, $V(G_{i}) \cap V(G_{j}) = \{s\}$.
Construction of $G_{i}$ and $T_{i}(s)$ for all $i$ takes a total time of $O(\sum_{i=1}^{l-1}(T_{SPT}(G_{i}))$ = $O(T_{SPT}(G))$.
$d_{G-v_{i}-T_{2}(v_{i})}(s,x)$ for any $x \in F(v_{i})$ is available in constant time from $T_{i}(s)$ of $G_{i}$.

Instead of computing $l-1$ SPTs, $T_{i}(s)$, for all $1 \leq i \leq l-1$, we compute one graph,
$\widetilde{G} = \bigcup^{l}_{i=1}G_{i}$, where $G_{i}$ is constructed as explained earlier.
$\widetilde{G}$ can be constructed from $G$ in $O(m+n)$ time. Single source shortest path tree rooted at $s$,
$\widetilde{T}_{s}$ of $\widetilde{G}$ is computed in $O(T_{SPT}(\widetilde{G})) = O(T_{SPT}(G))$ time. $d_{G-v_{i}-T_{2}(v_{i})}(s,x)$
for any $x \in F(v_{i})$ is available in constant time from $\widetilde{T}_{s}$ of $\widetilde{G}$.
Moreover, as $C''(v_{i}) \cap C''(v_{j}) = \emptyset, \forall (i \neq j)$,
the distances $d''_{G-v_{i}}(s,t)$ for all $v_{i}$ are available in linear time.

To compute $d'_{G-v_{i}}(s,t)$ for all $v_{i}$, we use the RSP-DAG. We use the vertex labeling on $T_{s}$ (as computed in Section~\ref{labeling}),
for a non-tree edge $(x,y)$, $(x,y) \in C'(v_{i})$
if and only if $label(x) < i$ and $label(y) > i$. In other words, for a non-tree edge $(x,y)$, if $label(x) = i$
and $label(y) = i + r$ for some $r > 1$, then $(x,y) \in C'(v_{j})$, for all $i < j < i + r$.

Hence, the crossing edges $C'(v_{i})$ will be part
of edge sets associated with the vertices $(i, i+r), r > 1$ in the RSP-DAG. After the BFS traversal on the RSP-DAG,
the minimum cost crossing edge (over $C'(v_{i})$) for $v_{i}$ is available in the edge set of the node $(i-1, i+1)$ in the RSP-DAG.
We do not need to perform the BFS traversal on the RSP-DAG again, because, the data populated during the BFS traversal for the edge replacement
paths suffices.

If we have the swap edge $(x, y)$ for the vertex $v_{i}$,
we can report the RSP in constant time as $s \leadsto x \rightarrow y \leadsto t$. Here $s \leadsto x$
is available from $T_{s}$ if $(x, y) \in C'(v_{i})$. It is constructed from SPTs $T_{s}$ and $\widetilde{T}_{s}$  if
$(x, y) \in C''(v_{i})$.

\begin{thm}
\label{thm3}
There is an algorithm for the
Node Replacement Path problem that runs in $O(T_{SPT}(G) + m + l^2)$ time using $O(m + l^2)$ space.
\end{thm}

\begin{proof}
$T_{SPT}(G)$ represents the time to compute SPTs $T_{s}$ and $T_{t}$.
Computing the distances $d''_{G-v_{i}}(s,t)$ for all $v_{i}$ takes $O(T_{SPT}(G) + m + n)$ time.
Computing $d'_{G-v_{i}}(s,t)$ for all $v_{i}$ using the RSP-DAG takes $O(m + l^2)$ time
and $O(m + l^2)$ space. Therefore the total time complexity is $O(T_{SPT}(G) + m + l^2)$.
Space complexity is $O(m + l^2)$ which is the space necessary to store the RSP-DAG.
\qed
\end{proof}

Using the linear time algorithms for SPT, for integer weighted graphs,
minor closed graphs our algorithm takes $O(m + l^2)$ time.

Proof of Theorem~\ref{thm1} is implied from theorems~\ref{thm2} and~\ref{thm3}.

\section{Conclusions}
\label{conc}
In this paper, we propose an $O(T_{SPT}(G) + m + l^2)$ time  and $O(m + l^2)$
space algorithm for the Replacement Paths problem. The asymptotic complexity of our algorithm matches
with that of Jay and Saxena~\cite{jay}. However,
our algorithm does not require external algorithms RMQ and integer sorting and it is easy to implement.
An advantage of our algorithm over~\cite{jay} and~\cite{linearrsp} is that, in a single iteration
both the edge and node replacement paths can be computed, whereas in~\cite{jay} and~\cite{linearrsp}
the algorithm has to be run independently for the edge and node replacement paths.
For dense graphs and graphs with small diameter our algorithm is optimal and matches that of~\cite{linearrsp}.
\bibliographystyle{splncs}
\bibliography{RSP_IPL}

\def\noopsort#1{}
\begin{thebibliography}{10}

\bibitem{dijkstra}
Dijkstra, E.W.:
\newblock A note on two problems in connection with graphs.
\newblock Numerische Mathematik \textbf{1} (1959)  269--271

\bibitem{fheaps}
Fredman, M.L., Tarjan, R.E.:
\newblock Fibonacci heaps and their uses in improved network optimization
  algorithms.
\newblock Journal of the ACM \textbf{34}(3) (1987)  596--615

\bibitem{thorup}
Thorup, M.:
\newblock Floats, integers, and single source shortest paths.
\newblock Journal of Algorithms \textbf{35}(2) (2000)  189--201

\bibitem{planarspt}
Henzinger, M.R., Klein, P., Rao, S., Sairam:
\newblock Faster shortest-path algorithms for planar graphs.
\newblock Journal of Computer and System Sciences \textbf{55}(1) (1997)  3--23

\bibitem{minorspt}
Tazari, S., Muller-Hannemann, M.:
\newblock Shortest paths in linear time on minor-closed graph classes, with an
  application to steiner tree approximation.
\newblock Discrete Applied Mathematics \textbf{157}(4) (2009)  673--684

\bibitem{fastermve}
Nardelli, E., Proietti, G., Widmayer, P.:
\newblock A faster computation of the most vital edge of a shortest path.
\newblock Information Processing Letters \textbf{79}(2) (2001)  81--85

\bibitem{mostvn}
Nardelli, E., Proietti, G., Widmayer, P.:
\newblock Finding the most vital node of a shortest path.
\newblock Theoretical Computer Science \textbf{296}(1) (2003)  167--177

\bibitem{vickreyprice}
Hershberger, J., Suri, S.:
\newblock Vickrey prices and shortest paths: What is an edge worth?
\newblock In: Proceedings of the 42nd IEEE symposium on Foundations of Computer
  Science. (2001)  252--259

\bibitem{kmve}
Malik, K., Mittal, A.K., Gupta, S.K.:
\newblock The $k$ most vital arcs in the shortest path problem.
\newblock Operations Research Letters \textbf{8} (1989)  223--227

\bibitem{jay}
Mahadeokar, J., Saxena, S.:
\newblock Faster replacement paths algorithms in case of edge or node failure
  for undirected, positive integer weighted graphs.
\newblock Journal of Discrete Algorithms \textbf{23} (2013)  54--62

\bibitem{linearrsp}
Lee, C., Lu, H.:
\newblock Replacement paths via row minima of concise matrices.
\newblock SIAM Journal on Discrete Mathematics \textbf{2}(1) (2014)  206--225

\bibitem{asksax}
Kare, A.S., Saxena, S.:
\newblock Efficient solutions for finding vitality with respect to shortest
  paths.
\newblock In: 6th IEEE International Conference on Contemporary Computing
  (IC3). (2013)  70--75

\bibitem{rmq}
Berkman, O., Schieber, B., Vishkin, U.:
\newblock Optimal doubly logarithmic parallel algorithms based on finding all
  nearest smaller values.
\newblock Journal of Algorithms \textbf{14}(3) (1993)  344--370

\end{thebibliography}
\end{document}